\documentclass[12pt]{article}

 \usepackage[margin=1.1in]{geometry}

\usepackage{amsfonts}
\usepackage{amsmath,amssymb}
\usepackage{graphicx,amscd,xypic}
\usepackage{bbding}
\usepackage{indentfirst}
\usepackage{mathrsfs, dsfont, bbm}
\usepackage[T1]{fontenc}
\usepackage{hyperref}
\usepackage{enumerate}
\usepackage{color}

\newtheorem{thm}{Theorem}[section]

\newtheorem{prop}[thm]{Proposition}
\newtheorem{coro}[thm]{Corollary}
\newtheorem{lemma}[thm]{Lemma}

\newtheorem{remark}{Remark}
\newenvironment{proof}{\hspace{0ex}\textsc{Proof}.\hspace{1ex}}{\hfill$\Box$\newline}
\DeclareMathOperator{\dd}{\mathrm{d\!}}
\DeclareMathOperator{\BE}{\mathbf{E}}
\DeclareMathOperator{\BP}{\mathbf{P}}
\DeclareMathOperator{\BX}{Z}

\DeclareMathOperator{\R}{\mathbb{R}}

\begin{document}
\title{An Optimal Consumption-Investment Model with Constraint on Consumption
}
\author{ Zuo Quan Xu\footnote{Department of Applied
Mathematics, Hong Kong Polytechnic University, Hong Kong. This author
acknowledges financial supports from Hong Kong Early Career Scheme (No. 533112), Hong Kong General Research Fund (No. 529711) and Hong Kong Polytechnic University. Email: \url{maxu@polyu.edu.hk}.
 }\ \ and Fahuai
Yi\footnote{School of Mathematical Sciences, South China Normal University, Guangzhou, China.
 The project is supported by NNSF of China(No.11271143 and No.11371155) and University Special Research Fund for Ph.D. Program of China (20124407110001 and 20114407120008). Email: \url{fhyi@scnu.edu.cn}.
} }
\date{November 22, 2013}
\maketitle
\begin{abstract}
A continuous-time consumption-investment model with constraint is considered for a small investor whose decisions are the consumption rate and the allocation of wealth to a risk-free and a risky asset with logarithmic Brownian motion fluctuations. The consumption rate is subject to an upper bound constraint which linearly depends on the investor's wealth and bankruptcy is prohibited. The investor's objective is to maximize total expected discounted utility of consumption over an infinite trading horizon. It is shown that the value function is (second order) smooth everywhere but a unique possibility of (known) exception point and the optimal consumption-investment strategy is provided in a closed feedback form of wealth, which in contrast to the existing work does not involve the value function. According to this model, an investor should take the same optimal investment strategy as in Merton's model regardless his financial situation.
By contrast, the optimal consumption strategy does depend on the investor's financial situation: he should use a similar consumption strategy as in Merton's model when he is in a bad situation, and consume as much as possible when he is in a good situation.
\\[3mm]
\end{abstract}

\section{Introduction}
\noindent
The publication of the monumental 1952 article \emph{Portfolio Selection} and the 1959 book of the same title by Harry M. Markowitz (1952, 1959) heralded the beginning of modern finance. To develop a general theory of portfolio choice, Samuelson (1969) and Merton (1969, 1971) initiated the study of dynamic optimal consumption-investment problems. The problem concerning optimal consumption-investment decisions involves the decisions of an investor endowed with some initial wealth who seeks to maximize the expected (discounted) utility of consumption over time. The decisions (called \emph{consumption-investment strategy}) are the consumption rate and the allocation of wealth to risk-free and risky assets over time. According to Merton (1975), studying this type of problems is the natural starting point for the development of a theory of finance.
\par
Samuelson and Merton's pioneering papers prompted researchers to contribute a considerable volume of new work on the subject in various directions. The literature has extensively covered the optimal consumption-investment problems in the financial markets that are subject to constraints and market imperfections. For example, the book authored by Sethi (1997) summarized the research conducted by Sethi and his collaborators on the optimal consumption-investment problems under various constraints such as bankruptcy prohibited, subsistence consumption requirement, borrowing prohibited, and random coefficients market. Fleming and Zariphopoulou (1991) considered the optimal consumption-investment problem with borrowing constraints. Cvitani and Karatzas (1992, 1993) considered the scenario in which the investment strategy of an investor is restricted to take values in a given closed convex set. Zariphopoulou (1994) considered the problem under the constraint that the amount of money invested in a risky asset must not exceed an exogenous function of the wealth, and bankruptcy is prohibited at any time. Elie and Touzi (2008) considered the optimal consumption-investment problem with the constraint that the wealth process never falls below a fixed fraction of its running maximum. Davis and Norman (1990), Zariphopoulou (1992), Shreve and Soner (1994), Akian, Menaldi, and Sulem (1996), and Dai and Yi (2009) considered proportional transaction costs in the study of optimal consumption-investment problems. These optimal consumption-investment models focus on the constraints on the wealth process and the investment strategy.
\par
Bardhan (1994) considered the optimal consumption-investment problem with constraint on the consumption rate and the wealth. The constraint is that the investor must consume at a minimal (constant) rate throughout the investment period, which is known as the subsistence consumption requirement, and must maintain their wealth over a low boundary at all times. However, in financial practice, an upper boundary constraint on the consumption rate typically exists in addition to the subsistence consumption requirement. An example of such scenario is an investment firm with cash flow commitments that is subject to regulatory capital constraints. No study in the extant literature has considered an upper boundary constraint on the consumption rate in the theory of optimal consumption-investment in intertemporal economies.
\par
Harry Markowitz, a Nobel laureate in economics, stated, ``It remains to be seen whether the introduction of realistic investor constraints is an impenetrable barrier to analysis, or a golden opportunity for someone with a novel approach; and whether progress in this direction will come first from discrete or from continuous-time models,'' in the foreword of the book by Sethi (1997). Research on the optimal consumption-investment problem that considers the upper constraint on the consumption rate is scant, although extensive research has been conducted on the problem involving other constraints, such as no bankruptcy or limits on the amount of money borrowed. Consequently, this research topic has not been sufficiently explored. This motivated us to investigate the optimal consumption-investment problems with constraint on consumption rate.
\par
In this paper, we consider a continuous-time consumption-investment model with an upper bound constraint on the consumption rate, which linearly depends on the amount of wealth of an investor at any time. The problem is considered in a standard Black-Scholes market with a risk-free and a risky asset over an infinite trading horizon. We make the usual assumption that shorting is allowed but bankruptcy is prohibited in the market. We will primarily use techniques derived from the theories of free boundary and viscosity solution in the field of differential equations to solve the problem (See e.g., Crandall and Lions (1983), Lions (1983), Fleming and Soner (1992), Dai, Xu and Zhou (2010), Dai and Xu (2011), Chen and Yi (2012)). As is well-known, the value function is the unique constrained viscosity solution of the associated Hamilton-Jacobi-Bellman (HJB) equation. Using this fact, we first prove that the viscosity solution of the equation is smooth everywhere but a unique possibility of (known) exception point. The detailed descriptions of an unconstrained and a constrained trading regions are then provided. Finally, we derive the optimal consumption-investment strategy in a closed feedback form of wealth. In contrast to the existing models, the optimal strategy explicitly given in our model does not involve the value function. The result shows that an investor should use a similar optimal consumption-investment strategy as in the unconstraint case when his financial situation is bad and should consume at the maximum possible rate when his situation is good.
\par
The paper is organized as follows. We formulate a continuous-time optimal consumption-investment model with constraint on the consumption rate in Section 2. A case without constraint is studies in Section 3. In Section 4, the associated Hamilton-Jacobi-Bellman equation to the problem is introduced and a case with homogeneous constraint is investigated. Using the techniques in the theory of viscosity solution, we show some properties of the value function of the problem in Section 5. The descriptions of an unconstrained and a constrained trading regions are provided in Section 6. Finally, we derive the optimal consumption-investment strategy in a closed feedback form of wealth in Section 7. We conclude the paper in Section 8.

\section{Programme Formulation}
\noindent
We consider a standard Black-Scholes financial market with two assets: a bond and a stock.
The price of the bond is driven by an ordinary differential equation (ODE)
\begin{align*}
 \dd P_t=rP_t\dd t,
\end{align*}
where $r$ is the risk-free interest rate. The price of the stock is driven by a stochastic differential equation (SDE):
\begin{align*}
 \dd S_t=\alpha S_t\dd t+\sigma S_t\dd W_t,
\end{align*}
where $\alpha$ is the mean return rate of the stock, $\sigma$ is the volatility of the stock, and
$W(\cdot)$ is a standard one-dimensional Brownian motion on a given complete probability space $(\Omega, \mathcal{F},\BP)$.
We denote by $ \{\mathcal{F}_t =\sigma(W_s,s\leqslant t),t>0\}$ the filtration generated by the Brownian motion.
The interest rate $r$, the mean rate of return $\alpha$, and the volatility $\sigma$ are assumed to be constant with $r>0$, $\sigma>0$, and $\mu:=\alpha-r>0$. There are no transaction fees or taxes and shorting is also allowed in the market.
\par
Let us consider a small investor in the market. The investor's trading will not affect the market prices of the two assets.
His trading strategy is self-financing meaning that there is no incoming or outgoing cash flow during the whole ivestment period. Then it is well-known that the wealth process of the investor is driven by an SDE:
\begin{align}\label{wealth}
\begin{cases}
\dd X_t=(r X_t+\pi_t\mu-c_t)\dd t+\pi_t\sigma \dd W_t,\\
\;\;X_0=x,
\end{cases}
\end{align}
where $x>0$ is the initial endowment of the investor, $\pi_t$ is the amount of money invested in the stock at time $t$, $c_t\geqslant 0$ is the consumption rate at time $t$. In this paper, we assume that no bankruptcy is allowed, that is
\begin{align}\label{constaint:nobankruptcy}
 X_t\geqslant 0, \quad \forall\; t>0,
\end{align}
 almost surely (a.s.). The target of the investor is to choose the best consumption-investment strategy $( c(\cdot), \pi(\cdot))$, which is subject to certain constraints specified below, to maximize the total expected (discounted) utility from consumption over an infinite trading horizon
\begin{align}\label{objective}
\textrm{maximize }\BE\left[\int_0^{\infty}e^{-\beta t}U(c_t)\dd t\right],
\end{align}
where $U: \R^+\mapsto \R^+$ is the utility function of the investor, which is strictly increasing, and $\beta>0$ is a constant discounting factor.
In this paper, we consider risk-verse investor only, this is equivalent to say $U(\cdot)$ is concave.
\par
The consumption-investment strategy $(c(\cdot), \pi(\cdot))$ is required to satisfy the following integrability constraint
\begin{align}\label{constaint:add1}
\BE\left[ \int_{0}^{T} e^{-\beta t} (\pi_t^2+c_t)\dd t\right]<\infty, \quad \forall\; T>0,
\end{align}
in which case, SDE \eqref{wealth} admits a unique solution $X(\cdot)$ satisfying
\begin{align*}
\BE\left[\int_{0}^{T} e^{-\beta t} |X_t| \dd t\right]<\infty, \quad \forall\; T>0.
\end{align*}
\par
If no other constraint on the consumption rate and investment strategy exists, problem \eqref{objective} becomes the classical Merton (1971)'s consumption-investment problem. However, in practice, constraint on the consumption rate always exists; for example, the consumption rate cannot be too low because an investor has basic needs, which are the minimal amount of resources necessary required for long-term physical well-being;
this is the so-called the subsistence consumption requirement. Another practical example is when the manager of a fund requests a fixed salary and a proportion of the managed wealth as a bonus. However, most of the wealth still belongs to the owner, and consequently, the manager cannot take excessive amounts from the total wealth.
These scenarios motivated us to consider an upper constraint on the consumption rate.
\par
In this paper, specifically, we assume that the consumption rate is upper bounded by a time-invariant linear function of wealth $X_t$ at any time:
\begin{align}\label{constaint:add3}
0 \leqslant c_t\leqslant k X_t+ \ell,\quad t\geqslant 0,
\end{align}
where $k$ and $\ell $ are nonnegative constants, at least one of which is positive.
\par
Denote the value function by
\begin{align}\label{objective0}
V(x):=\sup_{(c(\cdot),\pi(\cdot))}\BE\left[\int_0^{\infty}e^{-\beta t}U(c_t)\dd t\right].
\end{align}
where the consumption-investment strategy $(c(\cdot), \pi(\cdot))$ is subject to the constraints \eqref{constaint:nobankruptcy}, \eqref{constaint:add1} and \eqref{constaint:add3}.
\par
Same as Merton (1971)'s model, we focus on the constant relative risk aversion (CRRA) type utility function
\begin{align}\label{utilityfunction}
U(x)=\frac{x^p}{p},\quad x \geqslant 0
\end{align}
for some constant $0<p<1$. It is well-known that logarithmic utility function can be treated as a limit case of CRRA type utility function as $\log (x)=\lim\limits_{p\to 0}\frac{x^p-1}{p}$, so the results of this paper can be extended to cover logarithmic utility function.

\section{Merton' Model: A Case without Constraint}
\noindent
We first recall the well-known result of Merton (1971) for the scenario without constraint.
Define
\begin{align*}
\theta:=&\frac{\mu^2}{2\sigma^2(1-p)}>0,
 \end{align*}
 and
\begin{align*}
\kappa:=&\frac{\beta-p( \theta+r)}{1-p}.
 \end{align*}
\begin{thm}
 If $\kappa>0$ and there is no constraint on the consumption rate, i.e., $k=+\infty$ or $\ell=+\infty$, then the optimal consumption-investment strategy for problem \eqref{objective0} is given by
\begin{align*}
 (c_t, \pi_t)=&\left(\kappa X_t, \frac{ \mu }{\sigma^2(1-p)} X_t\right),\quad t\geqslant 0,
\end{align*}
and the optimal value is
\begin{align}\label{unconstraintV}
 V^{\infty}(x)=\frac{1}{p}{\kappa}^{p-1}x^p.
\end{align}
\end{thm}
\par
The optimal value $V^{\infty}(x)=\frac{1}{p}{\kappa}^{p-1}x^p$ will serve as an upper bound for the optimal value in scenarios with constraint.

\section{Hamilton-Jacobi-Bellman Equation}
\noindent
We adopt the viscosity solution approach in differential equations to solve problem \eqref{objective0}.
Let us start with proving some basic properties of the value function.
\begin{prop}\label{propertiesbasicV}
If $\kappa>0$, then the value function $V(\cdot)$ of problem \eqref{objective0} satisfies
 \begin{align}\label{upperboundV}
 V(x)\leqslant \frac{1}{p}{\kappa}^{p-1}x^p, \quad x>0.
 \end{align}
Moreover, $V(\cdot)$ is continuous, increasing, and concave on $[0,+\infty)$ with $V(0)=0$.
\end{prop}
\begin{proof}
Both the set of admissible controls and the optimal value of problem \eqref{objective0} are increasing in $\ell$ and consequently, an upper bound of the optimal value is given by the scenario $\ell=+\infty$. So the inequality \eqref{upperboundV} follows from \eqref{unconstraintV}.
\par
If the initial endowment of problem \eqref{objective0} is 0, then the unique admissible consumption-investment strategy is $(c(\cdot), \pi(\cdot))\equiv (0,0)$, so $V(0)=0$ and consequently,
 $V(\cdot)$ is continuous at 0 from the right by \eqref{upperboundV}. By the definition of $V(\cdot)$, it is not hard to prove its the concavity and monotonicity. We leave the details to the interested readers. The continuity of $V(\cdot)$ on $(0,+\infty)$ follows from its finiteness and concavity.
\end{proof}
\par
With this proposition, using the theory of viscosity solution in differential equations (See Crandall and Lions (1983), Lions (1983), Fleming
and Soner (1992)), we can prove that
\begin{thm}\label{HJBequation}
If $\kappa>0$, then the value function $V(\cdot)$ of problem \eqref{objective0} is the unique viscosity solution of its associated HJB
equation
\begin{multline}
 \label{pde1}
\beta V(x)-\sup_{\pi}\left(\tfrac{1}{2}\sigma^2\pi^2 V_{xx}(x)+\pi\mu V_{x}(x)\Big)-\sup_{ 0 \leqslant c\leqslant k x+ \ell}\Big(U(c)-cV_{x}(x)\right)-rxV_{x}(x)\\
=\beta V(x)+\frac{\mu^2}{2\sigma^2}\frac{V_x^2(x)}{V_{xx}(x)}+(c(x) -rx)V_x(x)-\frac{c^p(x)}{p}=0,\quad x>0,
\end{multline}
 in the class of increasing concave functions on $[0,+\infty)$ with $V(0)=0$, where
 \begin{align*}
c(x):=\min\left\{(V_x(x))^{\frac{1}{p-1}}, k x+\ell \right\},\quad x>0.
\end{align*}
 \end{thm}
\begin{proof}
Standard proof (See e.g., Zariphopoulou (1992, 1994)). We leave the details to the interested readers.
\end{proof}

\subsection{A Case with Homogeneous Constraint }
\noindent
We first consider the scenario with a homogeneous constraint on the consumption rate. The results will be useful in studying general scenarios in the following sections.
\begin{thm}\label{trivialcase}
If $k>0$, $ \ell=0$, and $\kappa>0$, then the optimal consumption-investment strategy for problem \eqref{objective0} is given by
\begin{align}\label{homocasec&pi2}
 (c_t, \pi_t)=&\left(\min\left\{ \kappa, k \right\}X_t, \frac{ \mu }{\sigma^2(1-p)} X_t\right), \quad t\geqslant 0,
\end{align}
and the optimal value is
\begin{align}\label{homocaseV}
 V(x)= \frac{\min\left\{\kappa, k \right\}^p}{p(\kappa(1-p)+\min\left\{ \kappa, k \right\}p)}x^p=
 \begin{cases}
 \frac{k^p}{p(\kappa(1-p)+kp)}x^p, &\quad k< \kappa;\\[4pt]
 \frac{1}{p } \kappa^{p-1}x^p, &\quad k\geqslant \kappa.
 \end{cases}
\end{align}
\end{thm}
\begin{proof}
Suppose $\kappa>0$. Let $V(\cdot)$ defined as in \eqref{homocaseV}. Then
 \begin{multline*}
c(x)=\min\left\{(V_x(x))^{\frac{1}{p-1}}, k x+\ell \right\}=\min\left\{(V_x(x))^{\frac{1}{p-1}}, k x \right\}\\
=\min\left\{ \frac{\min\left\{\kappa, k \right\}^{\frac{p}{p-1}}}{(\kappa(1-p)+\min\left\{ \kappa, k \right\}p)^{\frac{1}{p-1}}}, k\right\}x
=\min\left\{ \kappa, k \right\}x,
\end{multline*}
where we used the fact that
\begin{align*}
 \frac{\min\left\{\kappa, k \right\}^{\frac{p}{p-1}}}{(\kappa(1-p)+\min\left\{ \kappa, k \right\}p)^{\frac{1}{p-1}}}
 \geqslant  \frac{\min\left\{\kappa, k \right\}^{\frac{p}{p-1}}}{(\min\left\{ \kappa, k \right\}(1-p)+\min\left\{ \kappa, k \right\}p)^{\frac{1}{p-1}}}
 =\min\left\{ \kappa, k \right\}=k,
\end{align*}
when $k<\kappa$.
It is easy to check that $V(\cdot)$ and $c(\cdot)$ satisfy HJB equation \eqref{pde1}. Because $V(\cdot)$ is increasing and concave with $V(0)=0$, by Theorem \eqref{HJBequation}, $V(\cdot)$ is the value function of problem \eqref{objective0}.
 It is easy to verify that the value \eqref{homocaseV} is achieved by taking the consumption-investment strategy \eqref{homocasec&pi2}.
\end{proof}
\begin{coro} \label{k>kappa}
 If $k\geqslant \kappa>0$ and $\ell\geqslant 0$, then the optimal consumption-investment strategy for problem \eqref{objective0} is given by
\begin{align}\label{homocasec&pi3}
 (c_t, \pi_t)=&\left(\kappa X_t, \frac{ \mu }{\sigma^2(1-p)} X_t\right), \quad t\geqslant 0,
\end{align}
and the optimal value is
\begin{align}\label{homocaseV&kappa<k}
 V(x)=\frac{1}{p}{\kappa}^{p-1}x^p.
\end{align}
If $\kappa \leqslant 0$ and $\ell\geqslant 0$, then problem \eqref{objective0} is ill-possed, i.e., its optimal value is infinity.
\end{coro}
\begin{proof}
Both the set of admissible controls and the optimal value are increasing in $\ell$ and consequently, the scenario $\ell=+\infty$ gives an upper bound \eqref{unconstraintV}, $V(x)\leqslant V^{\infty}(x)=\frac{1}{p}{\kappa}^{p-1}x^p$.
It is easy to verify that the upper bound $\frac{1}{p}{\kappa}^{p-1}x^p$ is achieved by taking the consumption-investment strategy \eqref{homocasec&pi3}.
\par
If $\kappa$ goes down to $0$, then the optimal value $V(x)=\frac{1}{p } \kappa^{p-1}x^p$ goes to infinity.
Because the optimal value of problem \eqref{objective0} is decreasing in $\beta$, we conclude that $V(x)=+\infty$ if $\kappa\leqslant 0$ and $\ell\geqslant 0$.
\end{proof}
\\[10pt]
By Theorem \ref{trivialcase} and Corollary \ref{k>kappa}, we only need to study the scenario
\begin{align*}
\kappa>k>0,\quad \ell>0,
\end{align*}
 which are henceforth assumed unless otherwise specified.
\begin{remark}\label{k=0}
We will not study the scenario $\kappa>k=0$ and $\ell>0$, because it can be treated easily by a similar argument as follows. We will address this issue again at the end of the paper.
\end{remark}

\section{The Value Function: Continuity of the First Order Derivative}
\noindent
\begin{thm}\label{V:c1}
The value function $V(\cdot)$ of problem \eqref{objective0} is in $C[0,+\infty)\cap C^1(0,+\infty)$ if $r\leqslant k$;
and in $C[0,+\infty)\cap C^1((0,+\infty)\backslash \{x_e\})$ if $r>k$, where
\begin{align}\label{xe}
x_e:=\frac{\ell}{r-k},
\end{align}
 is the unique possibility of exception point, in which case, $V_x(x_e-)\leqslant (k x_e+ \ell)^{p-1}$ and $V(x_e)=\frac{1}{\beta p}(k x_e+ \ell)^{p}$.
\end{thm}
\begin{proof}
It is proved that $V(\cdot)\in C[0,+\infty)$ in Proposition \ref{propertiesbasicV}. Note $V(\cdot)$ is increasing and concave, so we can define the right and left derivatives as
\begin{align*}
V_x(x\pm ):=\lim_{\varepsilon\to 0+}\frac{V(x\pm \varepsilon)-V(x)}{ \pm \varepsilon}\geqslant 0,
\end{align*}
for all $x>0$. Moreover, both $V_x(\cdot\pm )$ are decreasing functions and $0\leqslant V_x(x+)\leqslant V_x(x-)<+\infty$ for all $x>0$.
\par
Now we show that $V(\cdot)$ is continuously differentiable on $(0,+\infty)\backslash\{x_e\}$. By Darboux's Theorem, it is sufficient to show that $V(\cdot)$ is differentiable on $(0,+\infty)\backslash\{x_e\}$, which is equivalent to $V_x(x-)= V_x(x+)$ for all positive $x \neq x_e$.
\par
Per absurdum, suppose $V_x(x_0+)< V_x(x_0-)$ for some $x_0>0$. Let $\xi$ be any number satisfying $V_x(x_0+)<\xi< V_x(x_0-)$.
Define
\begin{align*}
\phi(x)=V(x_0)+\xi (x-x_0)-N(x-x_0)^2,
\end{align*}
where $N$ is any large positive number. Then by the concavity of $V(\cdot)$,
\begin{multline*}
V(x)\leqslant V(x_0)+V_x(x_0-) (x-x_0)=\phi(x)+(V_x(x_0-)-\xi) (x-x_0)+N(x-x_0)^2\\
<\phi(x), \quad \textrm{ if } 0<x_0-x<\frac{1}{N}(V_x(x_0-)-\xi);
\end{multline*}
and
\begin{multline*}
V(x)\leqslant V(x_0)+V_x(x_0+) (x-x_0)=\phi(x)+(V_x(x_0+)-\xi) (x-x_0)+N(x-x_0)^2\\
<\phi(x), \quad \textrm{ if } 0<x-x_0<\frac{1}{N}(\xi-V_x(x_0+)).
\end{multline*}
Therefore, $V(x_0)=\phi(x_0)$ and $V(x)<\phi(x)$ in a neighbourhood of $x_0$. By Theorem \ref{HJBequation}, $V(\cdot)$ is a viscosity solution of HJB \eqref{pde1}, noting $\phi(\cdot)\in C^2(0,+\infty)$, so
\begin{multline*}
0 \geqslant
\beta \phi(x_0)-\sup_{\pi}\left(\tfrac{1}{2}\sigma^2\pi^2 \phi_{xx}(x_0)+\pi\mu \phi_{x}(x_0)\right)-\sup_{ 0 \leqslant c\leqslant k x_0+ \ell}\left(U(c)-c\phi_{x}(x_0)\right)-rx_0\phi_{x}(x_0)\\
=\beta V(x_0)-\frac{\mu^2\xi^2}{4\sigma^2N}-\sup_{ 0 \leqslant c\leqslant k x_0+ \ell}\left(U(c)-c\xi \right)-rx_0\xi
=\beta V(x_0)-\frac{\mu^2\xi^2}{4\sigma^2N}-g(\xi),
\end{multline*}
where
\begin{align*}
 g(\xi):=\sup_{ 0 \leqslant c\leqslant k x_0+ \ell}\left(U(c)-c\xi \right)+rx_0\xi,\quad 0<\xi<+\infty.
\end{align*}
Letting $N\to+\infty$, we get
\begin{align}\label{g(xi)>betav}
g(\xi)\geqslant \beta V(x_0),
\end{align}
for all $\xi\in (V_x(x_0+) ,V_x(x_0-))$.
\par%
On the other hand, because $V(\cdot)$ is concave, it is second order differentiable almost everywhere, and consequently, there exists a sequence $\{x_n: n\geqslant 1\}$ going up to $x_0$ such that $V(\cdot)$ is both first and second order differentiable at each $x_n$. By Theorem \ref{HJBequation},
\begin{multline*}
0 =\beta V(x_n)-\sup_{\pi}\left(\tfrac{1}{2}\sigma^2\pi^2 V_{xx}(x_n)+\pi\mu V_{x}(x_n)\right)-\sup_{ 0 \leqslant c\leqslant k x_n+ \ell}\left(U(c)-cV_{x}(x_n)\right)-rx_nV_{x}(x_n)\\
\leqslant \beta V(x_n) -\sup_{ 0 \leqslant c\leqslant k x_n+ \ell}\left(U(c)-cV_{x}(x_n)\right)-rx_nV_{x}(x_n)\\
=\beta V(x_n) -g(V_{x}(x_n))+r(x_0-x_n)V_{x}(x_n).
\end{multline*}
So \begin{align*}
g(V_{x}(x_n))\leqslant \beta V(x_n) +r(x_0-x_n)V_{x}(x_n).
\end{align*}
Note that $g(\cdot)$ is convex on $(0,+\infty)$, so it is continuous on $(0,+\infty)$. Hence
\begin{align}\label{g(v_x(x-))leqV}
g(V_x(x_0-))=\lim_{n\to+\infty }g(V_{x}(x_n))\leqslant \lim_{n\to+\infty }(\beta V(x_n) +r(x_0-x_n)V_{x}(x_n)) =\beta V(x_0).
\end{align}
Similarly, we have
\begin{align}\label{g(v_x(x+))leqV}
g(V_x(x_0+))\leqslant \beta V(x_0).
\end{align}
Noting that $g(\cdot)$ is convex on $(0,+\infty)$ and \eqref{g(xi)>betav},
\begin{align}\label{g(v_x(x-))g(v_x(x+))=V}
\max\{g(V_x(x_0-)), g(V_x(x_0+))\}\geqslant g(\xi)\geqslant \beta V(x_0),\quad V_x(x_0+) <\xi <V_x(x_0-).
\end{align}
By \eqref{g(v_x(x-))leqV}, \eqref{g(v_x(x+))leqV}, and \eqref{g(v_x(x-))g(v_x(x+))=V}, we conclude that
$g(\xi)=\beta V(x_0)$ for all $\xi\in [V_x(x_0+) ,V_x(x_0-)]$.
Note
\begin{multline*}
 g(\xi)=\sup_{ 0 \leqslant c\leqslant k x_0+ \ell}\left(U(c)-c\xi \right)+rx_0\xi\\
 =\begin{cases}
 U(k x_0+ \ell)-( k x_0+ \ell-rx_0)\xi, &\quad \textrm{ if } \xi\leqslant (k x_0+ \ell)^{p-1};\\
 \left(\tfrac{1}{p}-1\right)\xi^{\frac{p}{p-1}}+rx_0 \xi, &\quad \textrm{ if } \xi> (k x_0+ \ell)^{p-1}.
 \end{cases}
\end{multline*}
Therefore, $g(\cdot)$ is a constant on $ [V_x(x_0+) ,V_x(x_0-)]$ if and only if $k x_0+ \ell-rx_0=0$ and $V_x(x_0-)\leqslant (k x_0+ \ell)^{p-1}$. It can only happen in the scenario $r>k$, $x_0=x_e$ and $V_x(x_e-)\leqslant (k x_e+ \ell)^{p-1}$, in which case, $\beta V(x_e)=g(\xi)=\frac{1}{ p}(k x_e+ \ell)^{p}$. The proof is complete.
\end{proof}
\par
Although $V(\cdot)$ may not be differentiable at $x_e$ when $r>k$, we can still define $V_x(x_e-)$.
From now on, we denote $V_x(x_e ):=V_x(x_e-)$ unless otherwise specified.

\section{The Value Function: Properties}
\noindent
\begin{prop}\label{propertiesV}
The value function $V(\cdot)$ of problem \eqref{objective0} satisfies the following properties:
\begin{enumerate}[(a).]
 \item $V(x)/x^p$ is decreasing, and hence
\begin{align*}
 xV_x(x)\leqslant pV(x), \quad x>0;
\end{align*}
\item we have
\begin{align*}
 \frac{k^p}{p(\kappa(1-p)+kp)}x^p \leqslant V(x)\leqslant \frac{1}{p}{\kappa}^{p-1}x^p, \quad x>0;
\end{align*}
\item $V(\cdot)$ is strictly concave on $(0,+\infty)$ and $V_x(\cdot)$ is strictly decreasing on $(0,+\infty) $;
\item we have
\begin{align*}
 \frac{k }{ \kappa(1-p)+kp } {\kappa}^{p-1}x^{p-1} \leqslant V_x(x)\leqslant {\kappa }^{p-1}x^{p-1}, \quad x>0.
\end{align*}
\end{enumerate}
\end{prop}
\begin{proof}
We first consider the scenario $x\neq x_e$.
\begin{enumerate}[(a).]
\item
Let $V(x,\ell)$ denote the value function $V(x)$ with constraint \eqref{constaint:add3}.
Given the form of CRRA type utility function \eqref{utilityfunction}, the dynamics \eqref{wealth} and constraint \eqref{constaint:add3}, a standard argument can show that $V(\cdot, \cdot)$ is homogeneous of degree $p$, i.e.,
$$V(\lambda x, \lambda\ell)=\lambda^p V(x,\ell),\quad \lambda> 0.$$
Letting $\lambda=x^{-1}$,
\begin{align*}
V(1, x^{-1}\ell) =x^{-p} V(x,\ell),
\end{align*}
the property (a) follows from $V(1, x^{-1}\ell)$ is decreasing in $x$.
\item
The upper bound is given by \eqref{upperboundV}. The lower bound given by the scenario $\ell=0$ is \eqref{homocaseV}.
\item Note $V_x(\cdot)$ is decreasing by the concavity of $V(\cdot)$. Suppose it is not strictly decreasing. Then $V_x(x)=A$, $x\in(x_1,x_2)$ for some constant $A\geqslant 0$ and $(x_1, x_2)\subset (0,+\infty)$. It follows that $V_{xx}(x)=0$, $x\in(x_1,x_2)$. If $A=0$, because $V(\cdot)$ is concave and increasing, $V_x(x)=0$, $x\in(x_1,+\infty)$ which contradicts the property (b). Suppose $A>0$. Applying HJB equation \eqref{pde1},
\begin{multline*}
\beta V(x)-\sup_{\pi}\left(\tfrac{1}{2}\sigma^2\pi^2 V_{xx}(x)+\pi\mu V_{x}(x)\right)-\sup_{ 0 \leqslant c\leqslant k x+ \ell}\left(U(c)-cV_{x}(x)\right)-rxV_{x}(x)=0,\\
\quad x_1<x<x_2,
\end{multline*}
we get
\begin{align*}
\beta V(x)=\sup_{\pi}( \pi\mu A) +\sup_{ 0 \leqslant c\leqslant k x+ \ell} (U(c)-cA )+rxA=+\infty,
\quad x_1<x<x_2,
\end{align*}
 which contradicts the property (b).
\item The upper bound follows from the the properties (a) and (b).
Note $V(\cdot)$ is concave and apply the property (b),
\begin{multline*}
 V_x(x)\geqslant \frac{V(x+y)-V(x)}{y}\\
 \geqslant \frac{1}{y}\left( \frac{k^p}{p(\kappa(1-p)+kp)}(x+y)^p- \frac{1}{p}{\kappa}^{p-1}x^p\right), \quad x>0, \; y>0.
\end{multline*}
Let $y=\frac{\kappa-k}{k}x$ in the above inequality,
\begin{multline*}
 V_x(x) \geqslant \frac{k}{(\kappa-k)x}\left( \frac{k^p}{p(\kappa(1-p)+kp)}\left(\frac{\kappa}{k}\right)^px^p- \frac{1}{p}{\kappa}^{p-1}x^p\right)\\
 = \frac{k}{(\kappa-k)x}\left( \frac{ \kappa}{ \kappa(1-p)+kp} -1\right) \frac{1}{p}{\kappa}^{p-1}x^p\\
 = \frac{k}{ \kappa-k }\left( \frac{ \kappa p-kp}{\kappa(1-p)+kp} \right) \frac{1}{p}{\kappa}^{p-1}x^{p-1}\\
 =\frac{k}{\kappa(1-p)+kp} {\kappa}^{p-1}x^{p-1},\quad x>0.
\end{multline*}
Thus the property (d) is proved.
\end{enumerate}
For the scenario $x=x_e$, all the properties can be proved by a limit argument.
The proof is complete.
\end{proof}
\par
Define an unconstrained trading region $\mathcal{U}$ and a constrained trading region $\mathcal{C}$ as follows:
\begin{align*}
 \mathcal{U}:=&\{x>0: V_x(x)^{\frac{1}{p-1}} < k x+\ell\},\\
 \mathcal{C}:=&\{x>0: V_x(x)^{\frac{1}{p-1}}\geqslant k x+\ell \}.
\end{align*}
One of the main results of this paper is providing detailed descriptions of these two regions.
\par
It follows from Theorem \eqref{HJBequation} that
\begin{align}
 \label{pdeU}\beta V(x)+\frac{\mu^2}{2\sigma^2}\frac{V_x^2(x)}{V_{xx}(x)} -rx V_x(x)+\left(1-\tfrac{1}{p}\right)V_x(x)^{\frac{p}{p-1}}&=0, \quad x\in\mathcal{U};\\
 \label{pdeC} \beta V(x)+\frac{\mu^2}{2\sigma^2}\frac{V_x^2(x)}{V_{xx}(x)}+(k x+\ell-rx)V_x(x)-\tfrac{1}{p}(kx+\ell)^p&=0,\quad x\in\mathcal{C}.
\end{align}
\par
Define
\begin{align*}
 \eta:=\left(\frac{k}{\kappa(1-p)+kp}\right)^{\frac{1}{p-1}}\kappa >\kappa.
\end{align*}
\begin{prop}
We have
\begin{align}\label{subsetofC}
 \left(\frac{\ell}{\kappa -k},\;+\infty\right)\subseteq \mathcal{C},
\end{align}
and
\begin{align}\label{subsetofU}
 \left(0, \frac{\ell}{\eta -k} \right)\subseteq \mathcal{U}.
\end{align}
\end{prop}
\begin{proof}
By the property (d) in Proposition \ref{propertiesV}, we have
\begin{align*}
 V_x(x) \leqslant {\kappa}^{p-1}x^{p-1}, \quad x>0,
\end{align*}
and hence,
\begin{align*}
V_x(x)^{\frac{1}{p-1}}\geqslant \kappa x>kx+\ell, \quad \textrm{ if } x\in \left(\frac{\ell}{\kappa -k},\;+\infty\right),
\end{align*}
thus \eqref{subsetofC} follows.
\par
Similarly,
we have
\begin{align*}
V_x(x) \geqslant \frac{k }{ \kappa(1-p)+kp } {\kappa}^{p-1}x^{p-1}=(\eta x)^{p-1}, \quad x>0,
\end{align*}
and hence,
\begin{align*}
V_x(x)^{\frac{1}{p-1}}\leqslant \eta x<kx+\ell, \quad \textrm{ if } x\in \left(0,\;\frac{\ell}{\eta -k} \right),
\end{align*}
thus \eqref{subsetofU} follows.
\end{proof}
\begin{coro}\label{xeinC}
If $\kappa> r>k$, then $x_e\in\mathcal{C}$. If $r>\eta$, then $x_e\in\mathcal{U}$.
\end{coro}
\begin{proof}
If $\kappa> r>k$ , then
\begin{align*}
x_e=\frac{\ell}{r-k}>\frac{ \ell }{\kappa-k}.
\end{align*}
Similarly, notting $\eta>\kappa>k$, if $r>\eta$, then $r>k$, and
\begin{align*}
x_e=\frac{\ell}{r-k}<\frac{ \ell }{\eta-k}.
\end{align*}
The claim follows from the above result.
\end{proof}%
\par
Now we are ready to provide the detailed descriptions of the regions $\mathcal{U}$ and $\mathcal{C}$.
\begin{thm}\label{thm:regions}
If $\kappa\geqslant k+r$, then there exists a constant
\begin{align}\label{regionxstar}
 x^*\in \left [\frac{\ell}{\eta -k},\; \frac{\ell}{\kappa -k}\right]
\end{align}
such that
\begin{align}\label{regionsU}
 \mathcal{U}=(0,x^*),
\end{align}
and
\begin{align}\label{regionsC}
 \mathcal{C}=[x^*,+\infty).
\end{align}
\end{thm}
\begin{proof}%
In order to prove the claim, we first derive the formula of solution in the unconstrained region $\mathcal{U}$,
although the problem does not admit a closed form solution on $(0,+\infty)$.
\par
Let $\BX(\cdot): (c_1,c_2)\mapsto \mathcal{U}$ be determined by
\begin{align}\label{v_x}
V_x(\BX(c))=c^{p-1},\quad c_1< c< c_2.
\end{align}
By Corollary \ref{xeinC}, $x_e\notin \mathcal{U}$, so $V_x(\cdot)$ is continuous and strictly decreasing on $ \mathcal{U}$.
Thus $\BX(\cdot)$ is well-defined and strictly increasing. It follows that
\begin{align}\label{v_xx}
V_{xx}(\BX(c))\BX'(c)=(p-1)c^{p-2},\quad c_1< c< c_2.
\end{align}
Applying \eqref{v_x} and \eqref{v_xx}, equation \eqref{pdeU} becomes
\begin{align*}
\beta V(\BX(c))-\theta c^p \BX'(c)-rc^{p-1}\BX(c)+ \frac{p-1}{p}c^p=0,\quad c_1< c< c_2.
\end{align*}
differentiating with respect to $c$,
\begin{multline*}
\beta V_x(\BX(c))\BX'(c)-\theta (c^p \BX''(c)+pc^{p-1}\BX'(c))-r(c^{p-1}\BX'(c)\\
+(p-1)c^{p-2}\BX(c))+ (p-1)c^{p-1}=0.
\end{multline*}
Applying \eqref{v_x} again and eliminating $c^{p-2}$,
\begin{align*}
\beta c\BX'(c)-\theta (c^2 \BX''(c)+pc \BX'(c))-r(c \BX'(c)+(p-1) \BX(c))+ (p-1)c =0.
\end{align*}
Now we obtain an ordinary differential equation for $\BX(\cdot)$:
\begin{align}
 \label{pdex} & \mathcal{L} \BX=0,\quad c_1< c< c_2,
\end{align}
where
\begin{align*}
\mathcal{L} \BX:=-\theta c^2 \BX''(c)+(\beta-\theta p-r)c \BX'(c) +r(1-p) \BX(c)-(1-p)c.
 \end{align*}
\par
Now we are ready to prove \eqref{regionsU}.
Per absurdum, suppose that besides the original interval $(0,x^*)$, the unconstrained
region $\mathcal{U}$ contains another bounded interval by \eqref{subsetofC}. That is, there exist $x_1$ and $x_2$ such that
\begin{align*}
 x^*<x_1<x_2<+\infty,\quad (x_1,x_2)\subseteq \mathcal{U},\quad V_x(x_1)^{\frac{1}{p-1}}=kx_1+\ell,\quad V_x(x_2)^{\frac{1}{p-1}}=kx_2+\ell,
 \end{align*}
 where the last two identities are from the continuity of $V_x(\cdot)$ and Corollary \ref{xeinC}.
 Let $c_1=\BX^{-1}(x_1)$ and $c_2=\BX^{-1}(x_2)$. Then recalling \eqref{v_x},
 \begin{align}
\label{X(c1)} \BX(c_1) &=x_1=\frac{V_x(x_1)^{\frac{1}{p-1}} -\ell}{k}=\frac{V_x(\BX(c_1))^{\frac{1}{p-1}} -\ell}{k}=\frac{c_1-\ell}{k}>0,\\
 \label{X(c2)} \BX(c_2) &=x_2=\frac{V_x(x_2)^{\frac{1}{p-1}} -\ell}{k}=\frac{V_x(\BX(c_2))^{\frac{1}{p-1}} -\ell}{k}=\frac{c_2-\ell}{k}>0.
 \end{align}
Thus
 \begin{align*}
 c>c_1>\ell,\quad\textrm{ if } c_1<c<c_2.
 \end{align*}
 We now confirm $c\mapsto\frac{c-\ell}k$ is a supersolution of ODE \eqref{pdex} with boundary conditions \eqref{X(c1)} and \eqref{X(c2)}.
In fact, $c\mapsto\frac{c-\ell}k$ satisfies boundary conditions \eqref{X(c1)} and \eqref{X(c2)}, thus we only need to confirm $\mathcal{L} \left(\frac{c-\ell}k\right)\geqslant 0$. Note
 \begin{multline*}
\mathcal{L} \left(\frac{c-\ell}{k}\right)=(\beta-\theta p-r)\frac{c}{k}+r(1-p)\frac{c-\ell}{k}-(1-p)c \\
 =\left(\frac{\beta-\theta p-r+r(1-p)}{1-p}-k \right)(1-p)\frac{c}{k} -r(1-p)\frac{\ell}{k}\\
 =\left(\frac{\kappa(1-p)}{1-p}-k \right)(1-p)\frac{c}{k} -r(1-p)\frac{\ell}{k}\\
 =\left(\kappa-k \right)(1-p)\frac{c}{k} -r(1-p)\frac{\ell}{k}>\left(\kappa-k \right)(1-p)\frac{\ell}{k} -r(1-p)\frac{\ell}{k}\geqslant 0,
\end{multline*}
where we used the assumption $\kappa\geqslant k+r$ in the last inequality. Thus we proved $\frac{c-\ell}k$ is a supersolution of ODE \eqref{pdex} with boundary conditions \eqref{X(c1)} and \eqref{X(c2)}. Therefore, $$\BX(c)\leqslant \frac{c-\ell}k,\quad c_1\leqslant c\leqslant c_2,$$
and consequently, $ V_x(\BX(c))^{\frac{1}{p-1}}=c\geqslant k\BX(c)+\ell $ that contradicts $\BX(c)\in\mathcal{U}$, $ c_1< c< c_2$. Thus we proved \eqref{regionsU}. By the definitions of $\mathcal{U}$ and $\mathcal{C}$, \eqref{regionsC} follows immediately.
\par
The claim \eqref{regionxstar} follows from \eqref{subsetofC} and \eqref{subsetofU}.
\end{proof}

\section{The Value Function: Continuity of the Second Order Derivative and the Optimal Strategy}
 \noindent
\begin{thm}\label{mainthm}
Suppose $\kappa\geqslant k+r$. If $k\geqslant r$, then $V_{xx}(\cdot)\in C(0,+\infty)$. If $k<r$, then $V_{xx}(\cdot)\in C\big((0,+\infty)\backslash \{x_e\}\big)$,
where $x_e$ defined in \eqref{xe} is the unique possibility of exception point.
\end{thm}
Our main idea to prove the above result is to consider the dual function of the value function $V(\cdot)$. Making dual transformation
\begin{align}\label{def:v}
 v(y):=\max_{x>0}(V(x)-xy),\quad y>0.
\end{align}
Then $v(\cdot)$ is a finite decreasing convex function on $(0,+\infty)$.
Since $V_x (\cdot)$ is strictly decreasing, we denote the inverse function of $V_x(x)=y$ by
\begin{align}\label{x=I(y)}
I(y)=x.
\end{align}
By the property (d) in Proposition \ref{propertiesV}, $I(\cdot)$ is decreasing and mapping $(0,+\infty)$ to itself.
From \eqref{def:v},
 \begin{align}\label{v(y)exp1}
v(y)=[V(x)-xV_x(x)]\Big|_{x=I(y)}=V(I(y))-yI(y).
 \end{align}
Differentiating with respect to $y$,
 \begin{align}\label{v_y(y)exp1}
v_y(y) &=V_x(I(y))I'(y)-yI'(y)-I(y)=-I(y),\\
\label{v_yy(y)exp1}
v_{yy}(y)&=-I'(y)=-\frac 1{V_{xx}(I(y))},
 \end{align}
Inserting \eqref{v_y(y)exp1} into \eqref{v(y)exp1},
\begin{align*}
V(I(y))=v(y)-yv_y(y).
 \end{align*}
Making the transformation \eqref{x=I(y)}, applying \eqref{v(y)exp1}, \eqref{v_y(y)exp1}, \eqref{v_yy(y)exp1}, and $V_x(x)=y$, HJB equation \eqref{pde1} becomes
\begin{align}\label{pdev}
\beta (v(y)-yv_y(y))-\frac{\mu^2}{2\sigma^2}y^2 v_{yy}(y)+yd(y)+ryv_y(y)-\frac 1pd^p(y)=0,\quad y>0,
 \end{align}
where
\begin{align*}
d(y):=\min\Big\{ y^{\frac 1{p-1}},\;\ell-kv_y(y) \Big\}.
 \end{align*}
Equation \eqref{pdev} is quasilinear ODE, which degenerate at $y=0$. It follows that
\begin{align*}
v(y)\in C^2(0,+\infty)\cap C^{\infty}((0,+\infty)\backslash\{y^*\}),
 \end{align*}
 where $y^*=V_x(x^*)$ and $x^*$ is defined in Theorem \ref{thm:regions}.
 \par
Theorem \ref{mainthm} will follow from the following two propositions.
\begin{prop} \label{v_xxon(x,infty)}
Suppose $\kappa\geqslant k+r$.
Let $x^*$ be defined as in Theorem \ref{thm:regions}. If $k\geqslant r$, then $V_{xx}(\cdot) \in C[x^*+\infty)$. If $k<r$, then $V_{xx}(\cdot)\in C\big([x^*+\infty)\backslash \{x_e\}\big)$, where $x_e$ defined in \eqref{xe} is the unique possibility of exception point.
\end{prop}
\begin{proof}
By \eqref{v_yy(y)exp1}, to prove $V_{xx}(\cdot) \in C\big([x^*+\infty)\backslash \{x_e\}\big)$ is equivalent to prove $v_{yy}(y)>0$ for all $y\in (0,y^*]\backslash \{y_e\}$, where $y_e=V_x(x_e-)$.
\par
Suppose there exists a point $0<y_0< y^*$ such that $v_{yy}(y_0)=0$, which is the minimum value of $v_{yy}(\cdot)$ by the convexity of $v(\cdot)$. It follows that $v_{yyy}(y_0)=0$. Differentiating \eqref{pdev} with respect to $y$ yields
\begin{multline}
\beta ( -yv_{yy}(y))-\frac{\mu^2}{2\sigma^2}(2y v_{yy}(y)+y^2 v_{yyy}(y))+ \ell-kv_y(y)-kyv_{yy}(y)\\
+r(v_y(y)+yv_{yy}(y))+k( \ell-kv_y(y))^{\frac{1}{p-1}}v_{yy}(y)=0, \quad 0<y<y^*. \label{v_yy4smally}
 \end{multline}
Applying $v_{yy}(y_0)=0$ and $v_{yyy}(y_0)=0$, we get $(k-r)v_{y}(y_0)=\ell$ which is equivalent to $(r-k)x_0=\ell$ where $x_0=I(y_0)$.
Hence $x_0=x_e$ is the unique possibility of exception point which can only happen in the scenario $r>k$.
\par
It remains to show $v_{yy}(y^*)>0$. If $v_{yy}(y^*)=0$ which is the minimum value of $v_{yy}(\cdot)$. It follows that $v_{yyy}(y^*-)\leqslant 0$. By \eqref{v_yy4smally}, it follows $(k-r)v_{y}(y^*)\geqslant \ell$ which is impossible if $k\geqslant r$ because $v(\cdot)$ is decreasing. If $r>k$, then $(r-k) x^* \geqslant \ell$, $x^*\geqslant x_e$ which is also impossible because $x_e\in\mathcal{C}$ by Corollary \ref{xeinC}.
\end{proof}
\begin{prop}\label{v_xxon(0,x)}
Suppose $\kappa\geqslant k+r$.
Let $x^*$ be defined as in Theorem \ref{thm:regions}. Then $V_{xx}(\cdot) \in C(0, x^*]$.
\end{prop}
\begin{proof}
It is proved that $v_{yy}(y^*)>0$ in the proof of Proposition \ref{v_xxon(x,infty)}.
Suppose there exists a point $y_0> y^*$ such that $v_{yy}(y_0)=0$, which is the minimum value of $v_{yy}(\cdot)$ by the convexity of $v(\cdot)$.
It follows that $v_{yyy}(y_0)=0$. Differentiating \eqref{pdev} with respect to $y$ yields
\begin{align*}
\beta ( -yv_{yy}(y))-\frac{\mu^2}{2\sigma^2}(2y v_{yy}(y)+y^2 v_{yyy}(y))+ r(v_y(y)+yv_{yy}(y))+y^{\frac{1}{p-1}}=0,\quad y>y^*.
 \end{align*}
Applying $v_{yy}(y_0)=0$ and $v_{yyy}(y_0)=0$, we get $rv_{y}(y_0)=-y_0^{\frac{1}{p-1}}$ which is equivalent to $V_x(x_0)^{\frac{1}{p-1}}=rx_0$ where $x_0=I(y_0)$. However, by the property (d) in Proposition \eqref{propertiesV}, we have $V_x(x_0)^{\frac{1}{p-1}}\geqslant \kappa x_0>rx_0$.
The proof is complete.
\end{proof}
\par
Before proving the global continuity of the first order derivative of the value function, we recall a result in convex analysis.
\begin{lemma}\label{dual}
Let $h(\cdot)$ be a finite concave function on $(0,+\infty)$. Define its convex dual
\begin{align*}
 \widehat{h}(y):=\max_{x>0}(h(x)-xy),\quad y>0.
\end{align*}
Let $y_0=\inf\{y>0: \widehat{h}(y)<+\infty\}$. Then the following two statements are equivalent:
\begin{enumerate}
	 \item $\widehat{h}(\cdot) $ is strictly convex on $(y_0,+\infty)$.
 \item $h(\cdot) $ is continuous differentiable on $(0,+\infty)$.
\end{enumerate}
\end{lemma}
\begin{proof}
$''1 \Longrightarrow 2''$: Per absurdum, suppose $h(\cdot) $ is not differentiable at some $x_0>0$, then
\begin{align*}
h(x)- h(x_0)\leqslant y (x-x_0), \quad\forall\; x>0,
\end{align*}
for all $y\in [h_x(x_0+), h_x(x_0-)]$. Then it follows
\begin{align*}
 h(x)- y x\leqslant h(x_0)-y x_0 , \quad\forall\; x>0,
\end{align*}
and hence $\widehat{h}(y)= h(x_0)-y x_0$, $y\in [h_x(x_0+), h_x(x_0-)]$. This contradicts that $\widehat{h}(\cdot)$ is strictly convex.
Therefore, $h(\cdot) $ is differentiable. Because $h_x(\cdot) $ is increasing, by the Darboux's Theorem, $h_x(\cdot) $ is also continuous.
\par
$''2 \Longrightarrow 1''$:
Because $h(\cdot)$ is continuous differentiable on $(0,+\infty)$,
$$\widehat{h}(h'(x))= h(x)-h'(x)x,\quad x>0.$$
 For any $b>a>y_0$, we need to show $2\widehat{h}((a+b)/2)<\widehat{h}(a)+\widehat{h}(b).$
Let $0<x_1<x_2$ such that $b=h'(x_1)>h'(x_2)=a$.
Let $y$ satisfy $h'(y)=\frac{1}{2}(h'(x_1)+h'(x_2))=\frac{1}{2}(a+b)$, then $x_1<y<x_2$. It is sufficient to show
$$2\widehat{h}(h'(y))<\widehat{h}(h'(x_1))+\widehat{h}(h'(x_2)) ,$$
i.e.,
\begin{align}\label{inequality:h}
2h(y)-2h'(y)y<h(x_1)-h'(x_1)x_1+h(x_2)-h'(x_2)x_2.
\end{align}
Because $h(\cdot)$ is concave,
\begin{align*}
 h(y)-h(x_1)\leqslant (y-x_1)h'(x_1), \\
 h(y)-h(x_2)\leqslant (y-x_2)h'(x_2).
\end{align*}
If both of them are identities, then $h(\cdot)$ is linear on $[x_1,x_2]$, and $h'(x_1)=h'(x_2)$, a contradiction. So
\begin{multline*}
 2 h(y)-h(x_2)-h(x_1) < y (h'(x_2)+h'(x_1))-h'(x_1)x_1-h'(x_2)x_2\\
 = 2h'(y)y-h'(x_1)x_1-h'(x_2)x_2,
\end{multline*}
which is equivalent to the desired inequality \eqref{inequality:h}.
\end{proof}
\begin{coro}
Suppose $\kappa\geqslant k+r$.
The value function $V(\cdot)$ of problem \eqref{objective0} is in $C[0,+\infty)\cap C^1( 0,+\infty)$.
\end{coro}
\begin{proof}
It is proved that $v_{yy}(y)>0$ if $y\neq V_x(x_e-)$ in the proofs of Proportions \ref{v_xxon(x,infty)} and \ref{v_xxon(0,x)}.
This implies that $v(\cdot)$ is a strictly convex function on $(0,+\infty)$. Consequently, $V(\cdot)$ is continuous differentiable on $( 0,+\infty)$ by Lemma \ref{dual}.
\end{proof}
\par
To give an explicit optimal consumption-investment strategy for problem \eqref{objective0}, we derive the formula of $\BX(\cdot)$ in the unconstrained region $\mathcal{U}$, although we cannot obtain a closed form solution on $(0,+\infty)$, but it is adequate for our purpose.
\begin{prop}
Suppose $\kappa\geqslant k+r$.
Let $\BX(\cdot)$ be defined as \eqref{v_x} and $x^*$ be defined as in Theorem \ref{thm:regions}, then
\begin{align}\label{x=xc1}
 \BX(c)=\frac{1}{\kappa}c- \frac{1}{\kappa} ((k-\kappa) x^*+ \ell ) \left(\frac{c}{kx^*+\ell}\right)^{\lambda}, \quad 0<c\leqslant kx^*+\ell.
\end{align}
\end{prop}
\begin{proof}
Let $c^*=\BX^{-1}(x^*)$.
Because $V_x(\cdot)$ is in $ C^1( 0,+\infty)$ and \eqref{v_x},
\begin{align}\label{cstar}
c^*=V_x(\BX(c^*))^{\frac{1}{p-1}}=V_x(x^*)^{\frac{1}{p-1}}=kx^*+\ell.
\end{align}
The general solution of the corresponding homogeneous equation is $Bc^{\lambda} +\overline{B}c^{\overline{\lambda}}$, where $B$ and $\overline{B}$ are constants, $\lambda>\overline{\lambda}$ are two roots of function
 \begin{align*}
 f(\lambda)=\theta\lambda(\lambda-1)+(r-\beta+p\theta)\lambda+r(p-1).
 \end{align*}
 Note $f(1)=-\beta+p( \theta+r)<0$ and $f(+\infty)=+\infty$. It follows that $\lambda>1$ and $\overline{\lambda}<0$.
Note a particular solution to the inhomogeneous equation \eqref{pdex} is $\frac{1}{\kappa}c$.
 Thus the general solution to equation \eqref{pdex} is given by
\begin{align*}
 \BX(c)=\frac{1}{\kappa}c-Bc^{\lambda}-\overline{B}c^{\overline{\lambda}}, \quad 0<c\leqslant c^*.
\end{align*}
Because $\BX(0+)=0$ and $\overline{\lambda}<0$, we conclude that $\overline{B}=0$, and hence
 \begin{align*}
 \BX(c)=\frac{1}{\kappa}c-Bc^{\lambda}, \quad 0<c\leqslant c^*.
 \end{align*}
 By $\BX(c^*)=x^*$ and \eqref{cstar}, we obtain $B= \frac{1}{\kappa} ((k-\kappa) x^*+ \ell ) (kx^*+\ell)^{-\lambda}$ and
\begin{align*}
 \BX(c)=\frac{1}{\kappa}c- \frac{1}{\kappa} ((k-\kappa) x^*+ \ell ) \left(\frac{c}{kx^*+\ell}\right)^{\lambda}, \quad 0<c\leqslant c^*=kx^*+\ell.
\end{align*}
The proof is complete.
\end{proof}
\par
The main result of the paper is stated as follows.
\begin{thm}
Suppose $\kappa\geqslant k+r$. Let $x^*$ be defined as in Theorem \ref{thm:regions}.
The optimal consumption-investment strategy $(c^*(\cdot), \pi^*(\cdot))$ for problem \eqref{objective0} is given by a closed feedback form of wealth:
\begin{align*}
(c^*_t, \pi^*_t)=(c^*(X_t),\pi^*(X_t)),\quad t\geqslant 0,
\end{align*}
where
\begin{align*}
 c^*(x)=
 \begin{cases}
 \BX^{-1}(x),&\quad 0<x<x^*;\\
 kx+\ell,&\quad x\geqslant x^*,
 \end{cases}
 \quad \textrm{ and }\quad
 \pi^*(x)=\frac{ \mu }{\sigma^2(1-p)} x,\quad x>0,
\end{align*}
and $\BX^{-1}(\cdot)$ is the inverse function of $\BX(\cdot)$ defined in \eqref{x=xc1}.
\end{thm}
\begin{proof}
It is evident that
\begin{align*}
 c^*(x)=
 \begin{cases}
 V_x(x)^{\frac{1}{p-1}},&\quad 0<x<x^*;\\
 kx+\ell,&\quad x\geqslant x^*.
 \end{cases}
\end{align*}
We only need to show $V_x(x)^{\frac{1}{p-1}}=\BX^{-1}(x)$ which follows from \eqref{v_x}.
\end{proof}

\section{Concluding Remarks}
\noindent
As mentioned in Remark \ref{k=0}, the scenario $\kappa>k=0$ and $\ell> 0$ can be treated by our argument. In fact, in this scenario, both $\mathcal{U}$ and $\mathcal{C}$ are clearly intervals as $V_x(\cdot)$ is decreasing. Moreover, ODE \eqref{pdev} can be solved separately in the two regions. So we will not only have an explicit optimal consumption-investment strategy in a feedback form, but also have an explicit expression of the optimal value. We leave the details to the interested readers. As you may see, the problem is still open in the scenario $\kappa<k+r$. We will continuous work on this scenario and hope to fill the gap in the near future although the scenario is less likely to happen in real financial practice.



\end{document}